\documentclass{article}

\RequirePackage{amsthm,amsmath,amsfonts,amssymb}
\RequirePackage[authoryear]{natbib}
\RequirePackage[colorlinks,citecolor=blue,urlcolor=blue]{hyperref}
\RequirePackage{graphicx}
\usepackage[utf8]{inputenc}
\usepackage{amsmath}
\usepackage{xcolor}
\usepackage{graphicx}
\usepackage{amssymb,bm}
\usepackage{amsthm}
\usepackage{circuitikz}
\usepackage{subcaption}
\usepackage{natbib}   

\theoremstyle{plain}

\newtheorem{theorem}{Theorem}[section]

\newcommand{\bch}{\color{blue}\it}

\newcommand{\Mb}{\overline{M}}

\newcommand{\bt}{\widetilde{b}}
\newcommand{\thh}{\widehat{\theta}}

\newcommand{\eps}{\epsilon}
\newcommand{\epsb}{{\bm{\epsilon}}}
\newcommand{\betab}{{\bm{\beta}}}
\newcommand{\kappab}{{\bm{\kappa}}}
\newcommand{\betas}{\beta^\star}
\newcommand{\Gs}{G^\star}
\newcommand{\om}{\omega}
\newcommand{\omb}{{\bm{\omega}}}

\newcommand{\kb}{\bm{k}}
\newcommand{\hb}{\bm{h}}
\newcommand{\nb}{\bm{n}}

\newcommand{\logit}{\mbox{logit}}
\newcommand{\ind}{\overset{ind}{\sim}}

\newcommand{\Be}{\mbox{Be}}

\newcommand{\Bern}{\mbox{Bern}}
\newcommand{\Unif}{\mbox{Unif}}
\newcommand{\PG}{\mbox{PG}}
\newcommand{\Md}{\mbox{Mdn}}
\newcommand{\PT}{\mbox{PT}}
\newcommand{\DP}{\mbox{DP}}
\renewcommand{\AA}{\mathcal{A}}
\newcommand{\FF}{\mathcal{F}}
\newcommand{\ii}{^{(i)}}
\newcommand{\jj}{^{(j)}}
\newcommand{\ip}{{i^+}}
\newcommand{\im}{{i^-}}
\renewcommand{\ss}{^{(s)}}

\newcommand{\bepszs}{\beta_{\epsb0}\ss}

\title{Clustering and Meta-Analysis Using a Mixture of Dependent Linear Tail-Free Priors}
\author{Bernardo Flores \and Peter Müller}
\date{}

\begin{document}
\maketitle

\begin{abstract}
    We propose a novel nonparametric Bayesian approach for meta-analysis
with event time outcomes. The model is an extension of linear
dependent tail-free processes. The extension includes a modification to
facilitate (conditionally) conjugate posterior updating and a
hierarchical extension with a random partition of studies. The
partition is formalized as a Dirichlet process mixture.
The model development is motivated by a meta-analysis of cancer
immunotherapy studies. The aim is to validate the use of relevant
biomarkers in the design of immunotherapy studies. The hypothesis is
about immunotherapy in general, rather than about a specific tumor
type, therapy and marker. This broad hypothesis leads to a very
diverse set of studies being included in the analysis and gives rise
to substantial heterogeneity across studies.
\end{abstract}


\section{Introduction}

Meta-analysis is a widely used approach to synthesize data from
multiple clinical studies that are targeting the same hypothesis of
interest, allowing the pooling of information to produce more robust
inference (\cite{kelley_statistical_2012}). This can be done using
either aggregate or individual participant data. The latter permits to
obtain summary results directly from patient level data using
established statistical techniques; however, obtaining full sample
results from every clinical trial is rarely feasible, so an
aggregate approach is often taken (\cite{kelley_statistical_2012}). This
approach combines summary statistics reported for the individual studies to
estimate a global result. Classically, there are two ways this is
done; using either a random or a fixed effects model
(\cite{kelley_statistical_2012}). In a fixed effects meta-analysis, all
studies are assumed to be measuring the same (noisy) overall effect
$\theta$, so a regression problem is set up by $\thh_i = \theta +
\varepsilon_i$, where $\varepsilon_i$ is a residual with assumed known variance. 
In contrast, the random effects models sets up a hierarchical model
$\thh_i = \theta + \gamma_i + \varepsilon_i$ which takes into
account that the change in study-level conditions can introduce more
heterogeneity than what can be represented with a a homoskedastic
fixed effect model.

When the reported statistic is the median, many methods rely on first using the median to estimate the mean and variance, which introduces additional sources of errors and assumes the outcome variable to be symmetric; whereas the techniques that directly analyse the difference of medians use strong parametric assumptions that limit their flexibility \citep{mcgrath_estimating_2020}. 

Bayesian nonparametric priors (BNP) offer a natural way to relax
restrictive parametric assumptions by considering priors on random
distributions. More formally nonparametric Bayesian models can be
defined as prior probability models for infinite-dimensional parameter
spaces \citep{ghosal_fundamentals_2017}. The Dirichlet process (DP) introduced in \citep{ferguson_bayesian_1973} is arguably the
most popular such prior with support over discrete probability
measures. When used as a mixing measure for a mixture of Gaussian
kernels, the resulting mixture, known as DP mixture (DPM) of
normals, has full support over continuous
densities \citep{ghosal_fundamentals_2017}. The DPM prior, however is not the
simplest model on densities with full support. One can instead
think of constructing a random histogram by partitioning the sample
space and assigning random masses on each partitioning subset. If
these do not vary too wildly, continuity of the resulting probability
density can be achieved. This reasoning resulted in the development of tail-free processes,
with the seminal Pólya tree studied by Lavine in
\citep{lavine_aspects_1992} being the prime example
\citep{ghosal_fundamentals_2017}. Multivariate extensions of it that
allow the inclusion of covariates have been proposed
(\cite{trippa_multivariate_2011}, \cite{jara_class_2011}), but making the
computation scalable to large data sets remains problematic.  

Building on these models, we developed an extension of the linear
dependent tail-free process (LDTP) from \cite{jara_class_2011} to
explicitly deal with large heterogeneity while also permitting us to
directly use the reported medians and confidence intervals to do
inference. We refer to the proposed model as a Bayesian Nonparametric
Meta-Analysis (BNPMA). An alternative construction building on a more general dependent tailfree process using Gaussian process priors on the conditional splitting probabilities is developed in \cite{poli_multivariate_2023}.

This paper is organized as follows: in Section 2 we introduce the
motivating study together with a frequentist analysis of it. In
Section 3, we review the Pólya tree and introduce a BNP mixture of LDTP
priors, including an algorithm for posterior sampling. In
Section 4, we show a simulation study using data generated from a mixture of Weibull regressions. Finally, in Section 5, we implement the model with
the data presented in Section \ref{sec:immunoth}.

\section{A meta-analysis of cancer immunotherapy studies}
\label{sec:immunoth}
We analyze  data collected in \cite{fountzilas_correlation_2023},
which reports summary statistics from phase I and II trials for a type
of immunotherapy agents known as immune checkpoint inhibitors. These
therapies have shown great promise as new-generation cancer treatments
by incrementing the overall treatment outcomes in trials against
standard treatments or placebo \citep{fountzilas_correlation_2023}. The
goal was to evaluate the potential of biomarkers for the design of
immunotherapy studies and, more generally, to optimally match patients
with available treatments. For this the authors performed a
frequentist meta analysis using a random effects meta-regression. The
data is available in \cite{fountzilas_dataset_2023}.

Of the several considered response variables, we are interested in  progress-free survival (PFS), the time from initiation of treatment to the occurrence of disease progression or death. The data report PFS for $S=25$ studies, $s=1,\dots,25$. All selected studies report summaries stratified by marker status for some biomarker related to immunotherapy, defining a total of $I=53$ cohorts, $i=1,\dots,53$ (some studies are divided into more than two cohorts). In the following discussion we shall use $i^-$ and $i^+$ to index matching marker-negative and marker-positive cohorts of a study. We will use $s(i)$ to denote the study that
includes cohort $i$. The data are triples $(l,m,h)$ for each cohort
where $m$ is the median and $(l,h)$ a 95\% confidence interval for
it. The reported covariates were tumor type, therapeutic agent, line of treatment, treatment type (monotherapy, combination, etc.) and presence of biomarker.

Using the \texttt{R} package \texttt{metamedian}, we first implemented a frequentist meta-analysis based on both the median and the (estimated) mean. For the latter, we used the methodology developed in \cite{mcgrath_estimating_2020} to estimate the population mean and standard deviations from the reported inference summaries. We then estimated the difference of medians and means by fitting a random effects model. This resulted in the estimates shown in Table \ref{tab:freq}. Both median and mean were significant, though the variances were not stable due to the large heterogeneity in the reported effect sizes. This is clear in both forest plots, shown in Figure \ref{fig:freq}.

\begin{table}[bt]
\centering
\begin{tabular}{l|lll}
        & Estimate & SE     & P-Value          \\ \hline
Means   & 2.0525   & 0.5172 & \textless 0.0001 \\
Medians & 1.6358   & 0.4319 & 0.0002          
\end{tabular}
\caption{Results from the frequentist meta-analysis.}
\label{tab:freq}
\end{table}

\begin{figure}[bt]
    \centering
    \includegraphics[width=0.47\linewidth]{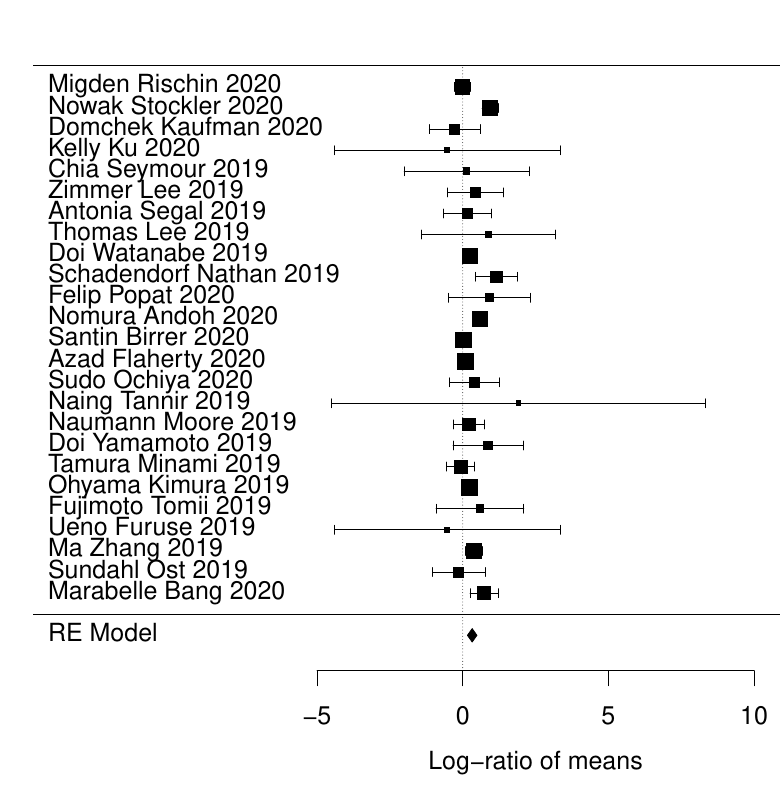}
    \includegraphics[width=0.47\linewidth]{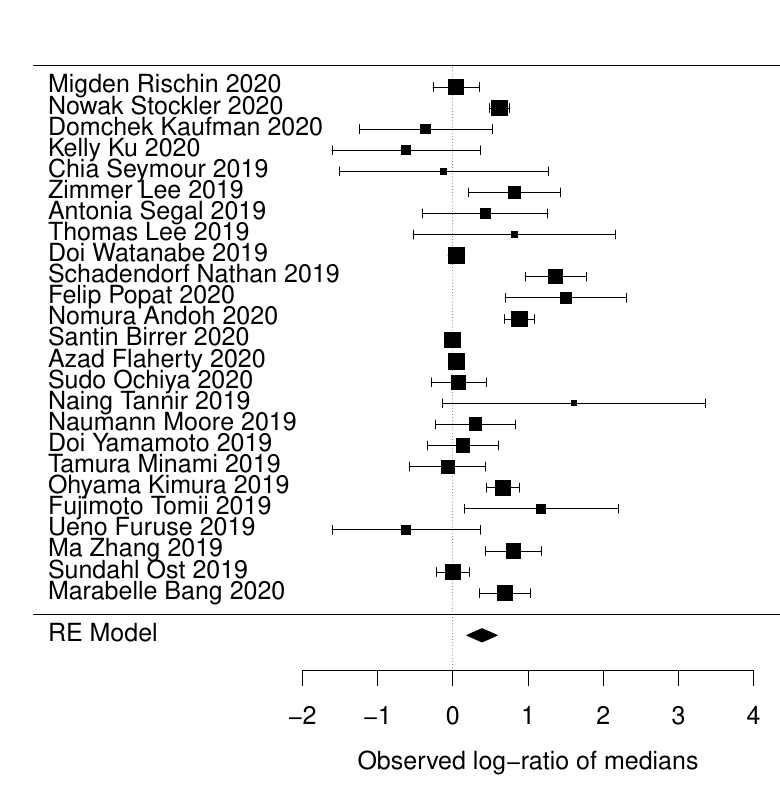}
    \caption{Forest plots of the estimated for the differences in mean (left) and median (right). The last line of the forest plots reports inference on the overall effect $\theta$.}
    \label{fig:freq}
\end{figure}

\section{Multivariate Pólya Tree}

\subsection{Univariate Pólya Tree}

The Pólya tree (PT) \citep{lavine_aspects_1992}  is a stochastic
process that generates random probability measures by recursively
partitioning the desired support and assigning random conditional
splitting probabilities to each branch of the resulting tree. The
implied random probability measure $P$ is often used as a
non-parametric Bayesian prior for unknown distributions such as
sampling models of random effects distributions(
\cite{diana_polya_2019}, \cite{zhao_mixtures_2009},
\cite{boeken_bayesian_2021}). 

The formal definition of a PT prior starts with partitioning the
sample space $B$ into $B=B_0\cup B_1$ and then
forming a tree of partitions
$\Pi$ by recursively splitting $B_{\eps_1\cdots \eps_m} =
B_{\eps_1\cdots \eps_m 0} \cup B_{\eps_1\cdots \eps_m
  1}$, where $\eps_i\in\{0,1\}$. Let $P$ denote the PT random
probability measure. At each split we define a random branching
probability 
\begin{align}
  Y_{\eps_1\cdots \eps_m 0}
  & \equiv P(B_{\eps_1\cdots \eps_m 0} \mid B_{\eps_1\cdots \eps_m} )
    \sim \text{Beta}(\alpha_{\eps_1\cdots \eps_m 0},\alpha_{\eps_1\cdots \eps_m 1})
\label{eq:Yeps}
\end{align}
and the complement $Y_{\eps_1\cdots \eps_m 1}  = 1-Y_{\eps_1\cdots \eps_m 0}$,
 implying 
for any
$\epsb=\eps_1,\cdots,\eps_m$
\begin{equation}\label{eq:polya}
P(B_\epsb) = \prod_{j=1}^m Y_{\eps_1\cdots \eps_{j}}.
\end{equation}

The class of distributions generated by this construction is quite
flexible, as proven in \cite{jara_class_2011}, with different choices
of the parameter sequence $\mathcal{A} = \{\alpha_\eps\}$ implying
densities with different characteristics. Assume that 
 $\alpha_\epsb=a_m$ for all 
$\epsb=\eps_1\cdots\eps_m$ at level $m$, that is,  
all subsets at the same level share the same $\alpha_\epsb$. 
 If the
sequence of parameters $\{a_m\}$ grows sufficiently fast,
specifically if
$
\sum a_m^{-1} < +\infty,
$
then it can be shown \citep{lavine_aspects_1992} that the corresponding Pólya tree generates continuous densities almost surely. 
Further, if for some $M>0$ and a distribution $G$ we set $\alpha_\eps
= M\,G(B_\eps)$, a DP with base measure $G$ and scale parameter $M$ is
recovered. 

Let $\Pi_m=\{B_{\eps_1\cdots\eps_m}\}$ denote the partitioning subsets
at level $m$ and $\Pi=\{\Pi_m,\; m=1,2,\ldots\}$, and let
$\AA=\{\alpha_{\epsb}\}$ denote the hyperprior parameters in
\eqref{eq:Yeps}.
We write $P \sim \PT(\Pi,\AA)$ to indicate a PT random measure.

To center a PT around a desired distribution $F_0$ two
different approaches can be used \citep{lavine_aspects_1992}.

One option is to start with fixed
$\Pi$, and then set $\AA$ by
\begin{equation}
    \alpha_{\eps 0} \propto F_0\left(B_{\eps 0} |  B_\eps\right)
    \label{aeps0}
  \end{equation}
 with proportionality across $\alpha_{\epsb0}, \alpha_{\epsb1}$,    
to ensure the desired prior expectation,  i.e.,
$E\{P(A)\}=F_0(A)$.
Alternatively, one can start by fixing $\AA$ 
with  $\alpha_{\epsb 0} = \alpha_{\epsb 1}$ for all $\epsb$ and then
set the partitioning subsets as the dyadic quantiles of $F_0$,
\emph{i.e.}  define the partition $\Pi_m$ at level $m$  as 
\begin{equation}
 \Pi_m =\{    \left[F_0^{-1}\left(k/2^m\right),
   F_0^{-1}\left((k+1)/2^m\right)\right], \quad k=0,\dots, 2^{m}-1\}.
 \label{eq:Pim}
\end{equation}
In the upcoming construction we use both to maintain a
common centering measure $F_0$ across cohorts, while allowing for
 a different partition sequence for each of them.
For levels $m=1,2$ we will use the first construction with fixed
$B_0, B_1, \ldots, B_{11}$, while for the remaining levels, $m>2$, we
will use a variation of the second construction, involving
dyadic splits of the respective parent set. That is,
$B_{\epsb0}, B_{\epsb1}$ are defined by splitting $B_\epsb$ at the
median of $F_0$ restricted to $B_\epsb$. Since $\Pi_1$ and $\Pi_2$ are
constructed differently, this defines a variation of
\eqref{eq:Pim}.

\subsection{Multivariate Pólya Tree}

The univariate PT model generates a random  probability measure
$P$ with full weak
support \citep{lavine_aspects_1992}, yet it does not naturally extend
to dependent families of distributions.
One construction to achieve
such an extension is introduced in \cite{jara_class_2011}
who proposed the linear dependent tail free process (LDTP)
 for a family $\FF=\{F_x;\; x \in X\}$ of random probability
measures indexed by covariates $x \in X$. 

 We describe the construction for the application to the event
time distributions $P_i$ for cohorts $i=1,\ldots,I$. Let then
$Y\ii_{\eps0}$ denote the splitting probabilities for $P_i$. The LDTP
defines $P_i=F_x$ for a cohort with covariates $x_i=x$ by setting up
$F_x$ with random splitting probabilities 
\begin{align}
\logit\left(Y\ii_{\eps 0}\right) &= \beta_{\eps 0}'x_i\label{eq:ldtp}\\
\beta_{\eps 0} &\sim N(0, \Psi_{\eps 0})\nonumber
\end{align}
for some net of covariance matrices $(\Psi_{\eps 0})_{\eps 0}$.
Similarly to \eqref{eq:polya}, model \eqref{eq:ldtp} defines a random probability measure $P_i$ by $P_i(B_\eps)= \prod_{j=1}^m Y\ii_{\eps_1\cdots \eps_{j}}$, if $\eps=\eps_1\cdots\eps_m$. It can be shown \citep{jara_class_2011} that the LDTP preserves some
of the properties of the PT model, namely (weak) posterior consistency
and conditions for obtaining (Lebesgue) absolutely continuous
densities, assuming that the partition tree corresponds to the dyadic
quantiles of some Borel measure supported on the real line. 


If for each level $m=1,\dots$ we assume a covariance matrix $\Psi_m =
2c\rho(m)^{-1} I$ for some constant $c$ and function $\rho$ such that
$\sum_{j=1}^\infty \rho(j)^{-1} < \infty$, then the LDTP generates
Lebesgue densities with probability one. The flexibility of this
logistic-normal construction lies in the inclusion of covariates with
an arbitrary design matrix, which allows for complex correlations
between the resulting trees.

One feature of this model that limits the applicability for the
desired meta analysis is the use of a common regression parameters
$\beta_\epsb$ across all studies, imposing a strong structure. This
assumption limits the ability to accommodate heterogeneity across
studies. We propose a solution for this problem by using a mixture of
LDTPs, produced by mixing over the coefficients of the logistic
regression in \eqref{eq:ldtp}. \
This translates to splitting the studies into different
groups and independently fitting a logistic regression in each of
them, which produces multimodality in the posterior distributions.
 Another limitation is the use of a common partitioning sequence
$\Pi$ across all cohorts $i$. We will introduce a variation of the
LDTP to address both limitations, to define an inference model
suitable for the desired meta-analysis problem. 

\section{BNP meta analysis}
\subsection{A mixture of LDTPs}

We introduce two extension of model \eqref{eq:ldtp}. First, we allow for study-specific regression coefficients $\beta^{(s)}$ to  fit the expected heterogeneity across studies. Second, we will pool strength across similar studies by introducing clusters of studies with shared common $\beta^{(s)} = \beta^*_k$ for all studies $s$ in cluster $k$.

Our proposal is to introduce mixing over the regression coefficients
$\beta$ with a Dirichlet process, which produces an unbounded number
of clusters of studies with shared coefficients $\beta^{(s)}$.
Recall that $s(i)=s$ indicates the study cohort $i$ belongs to; then the full hierarchical model for all $\eps$ is:

\begin{align}
  \logit\left(Y\ii_{\eps 0}\,|\,\beta_{\eps 0}^{(s)}\right)
  &= {\beta^{(s)}_{\eps 0}}'x_i + z_{\eps 0}^{(s)},~~ s = s(i)  \label{eq:dep}\\
 z_{\eps 0}^{(s)} &\sim \text{N}(0,\psi_{\eps 0}), \nonumber
\end{align}
 with random effects $z_{\eps0}\ss$, and
study-specific effects $\beta^{(s)}$ that arise from a discrete prior,
\begin{eqnarray}
\beta_{\eps 0}^{(s)} \, \mid \,G_{\eps 0} &\sim & G_{\eps 0} \nonumber \\ 
G_{\eps 0} &\sim & DP\left(N\left(0, c 2^{|\eps 0|} I\right), \alpha\right)\nonumber\\
 & & \psi_{\eps 0} \sim \text{InvGamma}(a,b);\; E[\psi_{\eps
                0}]=b/(a-1)
\label{bDP}                
\end{eqnarray}

The normal random effects together with the DP prior on $G_{\eps0}$
define a DP mixture of normals for the study-specific effects
$\beta\ss$.
The model induces the desired clustering on the logistic regression. 
This is because
$G_{\epsb0}$ is  discrete, implying positive probabilities of
ties. Let $\{\betas_{\epsb0,k}, k=1,\ldots,K_{\epsb0} \}$ denote the unique values of $\bepszs$.
The ties create clusters of studies with shared values
$\bepszs=\betas_k$.
This effect is best seen in the predictive distribution, which defines
the following recursive sampling of regression coefficients,
including the point masses for previously sampled coefficients. The
latter induces the ties.
See, for example,
\citep{ghosal_fundamentals_2017} for more discussion. We have
$$
p\left(\beta_{\eps 0}^{(s+1)} \mid \beta_{\eps 0}^{(s)},\dots,
  \beta_{\eps 0}^{(1)}\right)
\propto \alpha N\left(0, c
  2^{|\eps 0|} I \right) +  \sum_{r=1}^{s}
\delta_{\beta_{\eps 0}^{(r)}}, 
$$
where $ \delta_{\beta_{\eps 0}^{(s(i))}}$ is a Dirac measure centered at $\beta_{\eps 0}^{(s(i))}$.
Essentially, at each step the Dirichlet process samples with positive probability from one of the previous sampled values or it chooses a new one from the base measure. This causes different studies to share the same parameters, inducing clustering.
 Finally, in \eqref{bDP} we use a DP mixture with the additional
random effect, rather than only a DP prior. This particular structure of the DP mixture is chosen to induce
additional correlation for cohorts under the same study. 

\subsection{Study-specific partitioning trees and centering} 

A particularly attractive property of tail-free process priors
is that they easily allow to work with interval data. Recall that
the data reports for each cohort the quadruples
$q_i = (\ell_i, m_i, h_i, n_i)_{i=1}^N$,
where $m_i$ is a point estimate and
$(\ell_i,h_i)$ is some $(1-\alpha)$-confidence interval for median
PFS, and $n_i$ is the sample size.
Considering a single cohort, for the moment we drop the index
$i$ in the following discussion.
We then define $\Pi_1$ and $\Pi_2$ using subsets defined by the
quadruple $q_i$, 
\begin{align}
B_0 = F_0^{-1}((0,m)) \quad \text{and} \quad B_1=(m,\infty)\\
B_{00}=(0,l), \quad B_{01}=(l,m), \quad B_{10}=(m,h), \quad
  B_{11}=(h,\infty)
  \label{Bepsii}
\end{align}
 It can be argued that reporting $q_i$ implies counts of observed
PFS times in these intervals. 
For this argument, consider the first $n$ order statistics $X_{(1)},\dots,
X_{(n)}$ of the (unobserved) event times, and let $M$ be the true
median. Consider $Z\sim \text{Bin}(n, 0.5)$ that counts the number of
data points that fall to the left of $M$.
Next find the values $k, j$ such that
$P(Z\geq j)= \alpha/2$ and $P(Z< k)= \alpha/2$, that is, the
$\alpha/2$ and $1-\alpha/2$ quantiles of $Z$, respectively.  
Now, by definition $P(Z\geq j) = P\left(X_{(j)} < M\right)$ and $P(Z<
k) = P\left(X_{(j)} \geq M\right)$, implying that
$P\left(X_{(k)} < M < X_{(j)}\right) = 1-\alpha$,
i.e., $(X_{(k)}, X_{(j)})$ can be argued to be a $(1-\alpha)$
confidence interval for the unknown median $M$. We proceed assuming
that this is how the reported $(\ell_i,h_i)$ were determined (of
course, in reality most were probably based on a Kaplan-Meier curve).  

Let $n_\epsb$ denote the counts for the intervals in $\Pi_1$ and
$\Pi_2$. 
The argument implies 
\begin{equation}
  n_0=\lfloor n/2\rfloor, n_{00} = k \mbox{ and } n_{10}=j,
  \label{nij}
\end{equation}
with the
remaining counts at levels $m=1,2$ implied by the complements to
$n_i$. 

Finally, with \eqref{eq:dep} replacing the beta prior \eqref{eq:Yeps}
we lose the earlier mentioned easy prior centering of the PT
construction. Instead we add
for each cohort a cohort-specific intercept
$c_{\epsb0}\ii$ in \eqref{eq:dep} to ensure $E(Y_{\epsb0}\ii) =
F_0(B_{\epsb0}\ii \mid B_{\epsb}\ii)$, for cohort-specific
partitioning subsets $B_{\epsb}\ii$.
A minor complication arises from the fact that the nonlinear $\logit$
transformation does not preserve expectations. Instead we numerically solve an
optimization problem: 
$$
c_{\eps 0}^i = \arg
\min_{c} \left|F_0^{-1}(B_{\eps 0}\ii | B_\eps\ii ) -
  E\left[\logit^{-1}\left({c+\beta^{(s(i))}_{\eps 0}}'x_i +
      z_{\eps 0}^{(s(i))}\right)\right]\right| . 
$$
Using the intercept changes \eqref{eq:dep} to
\begin{align}
\logit\left(Y\ii_{\eps 0}\,|\,\beta_{\eps 0} \right) &= c_{\eps 0}^i+ z_{\eps 0}^{s(i)}+\beta_{\eps 0}'x_i \label{eq:npma}
\end{align}
 with fixed intercept $c_{\eps0}^i$. 

\subsection{A partially linear dependent tail-free process}

We define the prior \eqref{eq:dep} only for levels $m=1,2$
using the cohort-specific partitions \eqref{Bepsii}, for which the
counts \eqref{nij} allow informed posterior updating.
For the remaining levels the data does not provide any additional
information. We therefore continue the model beyond $m=2$ with
independent beta priors for $Y_{\epsb0}\ii$ using
$\alpha_{\epsb} = a_m$ for all $\epsb=\eps_1\cdots\eps_m$ at level
$m$. We use $a_m=2^m$ to obtain continuous densities. Recall that the partitioning subsets $B\ii_{\epsb0}$ beyond level 2
are defined 
by splitting the parent set $B\ii_{\epsb}$ at the median of the desired
centering measure $F_0$ resticted to $B\ii_{\epsb}$.

Adding this construction the full model takes the following form.
Let $m=|\epsb|$ denote the level in the tree, and let $s=s(i)$ denote
the study that includes cohort $i$. 
\begin{equation}
  \begin{array}{lrcl}     \label{eq:bnpma}
m=1,2: &  \logit \left(Y_{\epsb 0}\ii\right)  & =
  & c_{\epsb 0}\ii + z_{\eps 0}\ss + x_i'{\bepszs}\\
  & \bepszs & \sim & G_{\epsb0}, ~~ m=1,2\\
m>2: & Y_{\epsb0}\ii & \sim & \Be\left(c\cdot 2^m, c\cdot 2^m\right), ~~ m>2
\end{array}
\end{equation}
 with the DP prior on $G_{\epsb0}$ and the hyperprior as before. This completes the construction of the proposed inference model
for meta-analysis. We refer to \eqref{eq:bnpma} as nonparametric 
Bayesian meta analysis (BNPMA). 

\begin{theorem}
The marginal prior expectation of the random distribution $P\ii$ under the BNPMA equals $F_0$. 
\end{theorem}
\begin{proof}

In the following argument we drop the $i$ sub and suuper indices and let $s=s(i)$ denote the study cohort $i$.

Start by noting that the splitting probabilities under the Dirichlet process with a Gaussian base-measure are marginally independent logit-normal random variables. Then

$$
c_{\eps 0}+x'\beta_{\eps 0}^s |  z_{\eps 0}^{s} \sim  N\left(c_{\eps 0}+ z_{\eps 0}^{s}, 2^{|\eps 0|}x ' x\right).
$$

Independence accross $\eps$ implies that the process is (marginally) tail-free. Then to achieve the desired centering all we need \citep{ghosal_fundamentals_2017} is for all $\eps$ and $i$
$$
E\left[P_i\ii(B_\eps)\right] = F_0(B_\eps)
$$

The moments of the inverse-logit distribution are not available in closed form \citep{holmes_moments_2022}, but can be calculated through numerical integration. But recall that $c_{\eps 0}$ was fixed such that the expected value of each $\logit^{-1}\left(c_{\eps 0}+ z_{\eps 0}^{s}+x'\beta_{\eps 0}^s \right)$ matches $F_0(B_{\eps 0})$, which centers the process in the desired distribution.
\end{proof}

\section{Posterior inference}

We implement posterior MCMC simulation for $P_i$ for all
cohorts $i=1,\ldots,I$. Posterior inference on $P_i$ then implies posterior
inference on the median $M_i=\Md(P_i)$, and the difference of  log medians
$\log(M_\ip)-\log(M_\im)$ for any pair of cohorts $\ip,\im$ which
define matching marker-positive and marker-negative cohorts in the
same study $s=s(i)$.  
%
Since the splitting probabilities $Y_{\epsb0}\ii$ 
are all independent  across $\epsb$, we can treat each $\epsb$ as
a separate inference problem, and implement separate instances of posterior MCMC simulation.

Consider thus the task of estimating the branching probability for a
given $\epsb$.
 We focus on the first two levels, $m=1,2$ only. Given $Y_{\epsb0}\ii$
for the first two levels the remaining splitting probabilities are
easily imputed from the prior -- there is no information in the data
and therefore no posterior updating beyond level $m=2$.

\paragraph*{ Updating \bch $G_{\epsb0}$.}

The random probability measures $G_{\epsb0}$ are {\em a
  posteriori} independent and can be updated in parallel.
For the upcoming discussion we therefore drop the $_{\epsb0}$ index. 
We use a finite DP approximation
\citep{ishwaran_gibbs_2001}.
Let then $G = \sum_{h=1}^H w_h \delta_{b_h}$ denote a finite truncation
of $G$ after $H$ terms.
 The unique values $\betas_k$ that were introduced in the comments
following \eqref{bDP} are tied with some of the $b_h$, but
not necessarily in the same order. For the upcoming discussion it is
convenient to instead index the unique values with $h$ for the
corresponding $b_h$, noting that there can be some
values $b_h$ with $\beta\ss \ne b_h$ for all $s$.
Letting $C^s_h=\{s:\; \beta\ss=b_h\}$ defines a partition of
$\{1,\ldots,S\} = \bigcup_{h=1}^H C^s_h$,  allowing for empty
clusters $C^s_h=\emptyset$ when $b_h \ne \beta\ss$ for all $s$.
Similarly we define $C_h=\{i:\; \beta^{(s(i))}=b_h\}$ as the
corresponding partition of cohorts.
For an alternative characterization of the partition we use cluster
membership indicators $h_i=h$ if $i \in C_h$.

The DP prior $G \sim \DP(\Gs, \alpha)$ can be constructively defined 
as $w_h=v_h \prod_{\ell<h} (1-v_h)$ with $v_h \sim \Be(1,\alpha)$
and $b_h \sim \Gs$, i.i.d., a priori. This is known as the
stick-breaking construction \citep{sethuraman_constructive_1994}. 
Posterior updating for the weights is carried out by updating
the stick-breaking components $(v_h)$ using the complete
conditional posterior
$v_h \mid \ldots
\sim \Be\left(m_h, \alpha + \sum_{j=h+1}^H m_j\right)$, where $m_h$
is the number of cohorts with $\beta\ii=b_h$. 

\paragraph*{ Logistic regression parameters}~
 For empty clusters $b_h$ is resampled using the base measure
$\Gs$.
For non-empty clusters $C_h$ we proceed as follows.
For each cohort $i$ let $n_i$ denote the number of observed
event times in $B\ii$ (recall that we are dropping the $\epsb0$
indices) given in \eqref{nij}, and let $N_i$ denote the count in the
corresponding parent set. Then the conditional posterior for
$b_h$ given the observed counts $n_i, N_i$ is 
\begin{equation}
p(b_h \mid \hb, \nb) \propto
\Gs(b_h) \cdot 
\prod_{i \in C_h}
    \frac{{\exp(x_i'b_h)}^{n_i} }
         {{(1+\exp(x_i'b_h))}^{N_i} }
    \label{pbk}
  \end{equation}
We include the random effects $z\ss$ in $b_h$, replacing $b_h$ by $\bt_h=b_h \cup \{z\ss:\; s \in
  C^s_h\}$ and $x_i$ by $\tilde{x}_i = \begin{bmatrix}x_i & \mathcal{Z}_i\end{bmatrix}$, where $\mathcal{Z}_i$ is the design vector for the random effects. Hence in the following discussion we will group together $z\ss$ with $b_h$ and update them simultaneously. 

We use the data augmentation method of \cite{polson_bayesian_2013} to
implement sampling of \eqref{pbk}.  The scheme involves a latent
Pólya-gamma distributed random variables $\om_i$, with complete
conditionals for $\om_i$ being a scaled Polya-gamma and for $\tilde{b}_h$
being normal linear regression posterior distribution. See \cite{polson_bayesian_2013} for
details. 

\paragraph*{ Gibbs sampler}

We define a Gibbs sampler using the algorithm in \cite{ishwaran_gibbs_2001} by iterating over the following steps. Here $\ldots$ in the conditioning sets stands for all other parameters and the data,
$h=h_i$ when the relevant cohort $i$ is understood from the
context, 
$\hb=(h_1,\ldots,h_n)$,
$m_h = |C_h|$ is the number of cohorts in cluster $h$,
$\nb=(\nb_i,\; i=1,\ldots,I)$,
$\omb=(\om_i,\; i=1,\ldots,I)$ and
$\PG(a,b)$ indicates a Pólya-Gamma distribution with parameters
$(a,b)$.

We denote the design matrix (including the random effects)
for a non-empty cluster $C_h$ as $X_h$, and the diagonal matrix constructed from
the PG latent variables belonging to $C_h$ as $\Omega_h$. Let $\tilde{b}_h=(b_h, z^{(s)}; s\in C_h^s)$ and $\tilde{x}_i$ the $i$-th row of the extended design matrix with the random effects included and let then the prior variance of $\tilde{b}_h$ be  
$$
\Sigma_0 = \begin{pmatrix}
    c\,2^{-m} I & \bf{0}\\ \bf{0} &\psi I
\end{pmatrix},
$$ for $\psi >0$.

In step 1, let $\kappa_i = n_i-N_i/2$ and $\kappab_h=(\kappa_i,\; i \in C_h)$. 

\begin{enumerate}
\item $p(\tilde{b}_h \mid \hb, \omb, \cdots) =
  N\left(\mu_\omega, V_\omega\right)$ where $V_\omega = (\tilde{X}_h'\Omega_h \tilde{X}_h+\Sigma_o^{-1})^{-1}$
  and $\mu_{\omega_h} = V_{\omega_h}(\tilde{X}_h'\kappab_h+\Sigma_o^{-1}\mu_o)$
\item $p(\om_i \mid \cdots \kb, \betas_{h_i},\ldots)
  = \PG(1, \tilde{x}_i' \tilde{b}_{h_i})$
 \item $p(h_i \mid \ldots) =
   \sum_{h=1}^H p_h\delta_{h}$ with
   $p_h \propto \left(w_h\,
     \text{Binom}(n_i \mid N_i,
     \logit^{-1}( \tilde{b}_h' \tilde{x}_i ) \right)_k$ 
\item $w_h = (1-V_1)(1-V_2)\cdots (1-V_{h-1})V_h$, where $V_h
  \ind \Be\left(m_h, \alpha + \sum_{\ell= h+1}^H m_\ell\right)$ 
\item Recall the hyperprior $\psi \sim \text{InvGamma}(a,b)$, implying the familiar conjugate InvGamma conditional posterior distribution.
\end{enumerate}

\section{Simulation study}
We set up a simulation generating data for $S=30$ studies with $I=60$
cohorts, including one marker-positive cohort $\ip$ and one marker-negative 
cohort $\im$ for each study.
For each study, both, $\ip$ and $\im$, share the same covariates, a
normal variable $x_{i1} \sim N(5,1)$ 
and a binary Bernoulli random variable $x_{i2} \sim \Bern(0.7)$. Next we generated study-specific random effects $\gamma_s \sim
\Unif(-0.3, 0.3)$ for each study. With these variables we then simulated data from a mixture of two
Weibull regressions (c.f. \citep{lemeshow_applied_2008}), choosing each
of the two regression models with equal probability.
Let $k_i \sim \Bern(0.5)$ denote a group assignment for cohort $i$.
In summary, assuming $n_i=50$ patients for each study, the generative
model for simulated event times $t_{ij}, j=1,\ldots,n_i=50$ patients per
cohort is as follows.
\begin{align}
  \log(t_{ij}) &= \betab_{k_i}' x_i + \gamma_i + \log w_{ij}
  \mbox{ with } w_{ij} \sim \text{Exp}(1).
                 \nonumber \\
\gamma_s &\sim \text{Unif}(-0.3, 0.3) 
           \label{M0}
\end{align}
We use $\betab_1 = (0.4, 0.2)$ and $\betab_2 = (0.7, 0.5)$.
Based on the simulated patient-level data we then consider
Kaplan-Meier plots, and read off estimates $m_i$ of the median and a
corresponding 95\% confidence interval $(\ell_i,h_i)$.
We record the quadruple $q_i=(\ell_i,m_i,h_i,n_i)$ as the hypothetical
data in the simulation.
Note that in this construction the simulation truth is different from
the analysis model that is assumed under the proposed BNPMA approach.
Figure \ref{fig:surv_sim} shows the simulation truth for each of the
60 simulated cohorts, separated by true cluster assignment. 

Conditional on these data we then implement Markov chain Monte Carlo (MCMC)
posterior simulation as described before.
We simulate 50,000 iterations to evaluate posterior summaries,
including posterior expectations for the median $M_i$ of the event time
distributions $P_i$.
We compare the posterior estimates with the simulation truth $M_i^o$
under \eqref{M0}.
The latter are calculated from the Weibull
regression as $q_{\frac 12} = \log\left(\frac 12\right) \exp{\betab_{k_i}'x +
  \gamma_i}$.
The earlier, i.e., the posterior means, are obtained from the 
posterior Monte Carlo samples by 
evaluating for each iteration of the MCMC simulation
the branching probabilities and then constructing a histogram
as in \eqref{eq:polya} to determine the implied median $M_i$
corresponding to each posterior sample. Finally, the Markov chain
Monte Carlo average of $M_i$ across iterations evaluates the posterior
expectation $\Mb_i = E(M_i \mid data)$.
Figure \ref{fig:sim}(a) compares posterior estimates $\Mb_i$ versus the
simulation truth $M_i^o$. The accumulation of the scatter plot around
the identity is evidence that posterior inference under BNPMA was
indeed able to reocver the true effects under these realistic sample
sizes and effect.
\begin{figure}[bt]
  \begin{tabular}{cc}
    \includegraphics[height=1.8in]{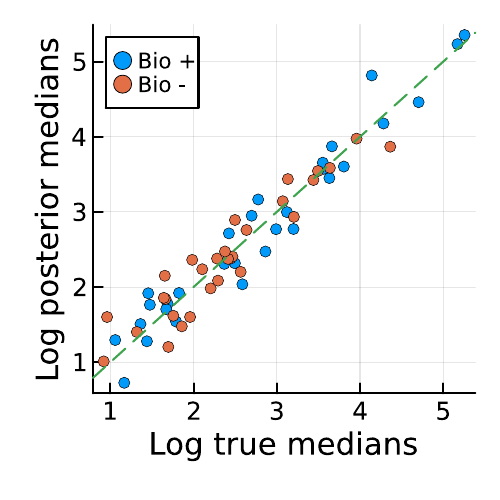}
    & \includegraphics[height=2.3in]{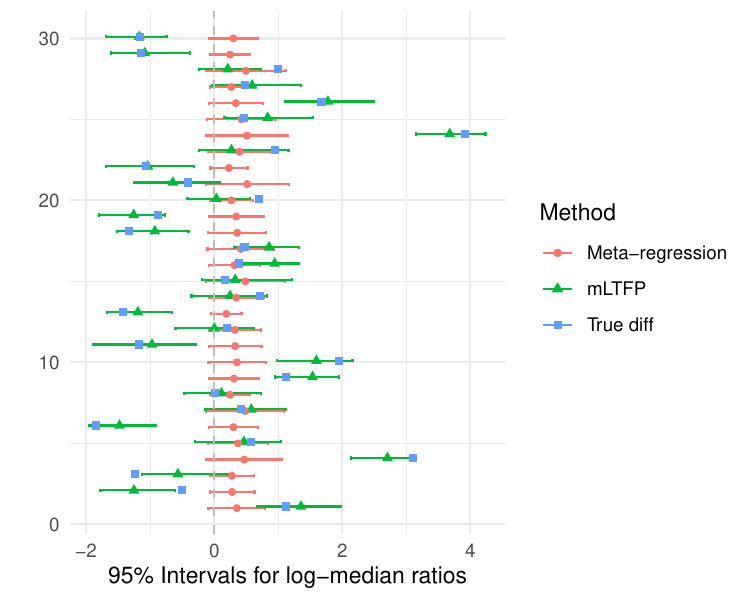}\\
    (a) & (b)
  \end{tabular}          
    \caption{Posterior inference on the medians $M_i$ in the
    simulation study. Panel (a) 
    plots posterior expected medians $\Mb_i=E(M_i \mid data)$ vs true
    medians $M_i$ for all cohorts (in log-log scale).
    Panel (b) shows a forest plot of the estimated log-ratio of
    medians and corresponding confidence and credible interval using
    classical meta-analysis using a random effects model (red, with bullets) and the BNPMA (green with
    triangles), respectively.  For reference, the simulation truth is also
    shown (blue, with little boxes).}
    \label{fig:sim}
\end{figure}

Finally, the forest plot in Figure \ref{fig:sim}(b) shows the
estimated log-ratio of medians for pairs $(\ip,\im)$ of
marker-positive and marker-negative cohorts under each study.
The figure shows estimates from a meta regression using both, the
proposed BNPMA and an established frequentist inference method
\citep{mcgrath_meta-analysis_2020}. For the latter first the mean and standard deviations are obtained from the medians and intervals using the estimators developed in \cite{luo_optimally_2018}; after which a random effects model with moderators is fitted.

The frequentist model shows shrinkage towards an overall effect,
with little variation across its estimates. We estimated a random effects model using the function \texttt{rma} from the package \texttt{metafor} with default parameters, which induced the strong shrinkage seen in the figure. However, note the horizontal scale of the figure, to accommodate the fit of the several outlier studies, which makes the shrinkage under the meta-regression appear more extreme. BNPMA instead induces more heterogeneity, still with shrinkage of the estimates towards the same overall effect, but allowing the fitting of the several outlier studies.

\begin{figure}[bt]
    \centering
    \includegraphics[width = 0.3\linewidth]{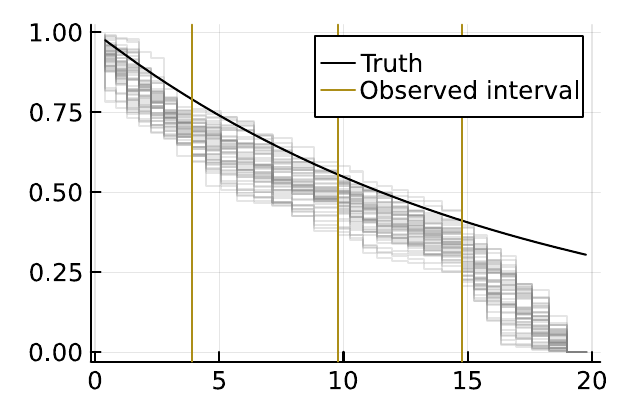}
    \includegraphics[width = 0.3\linewidth]{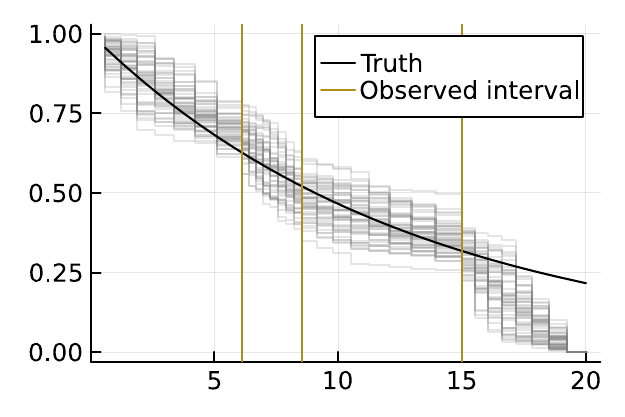}
    \includegraphics[width = 0.3\linewidth]{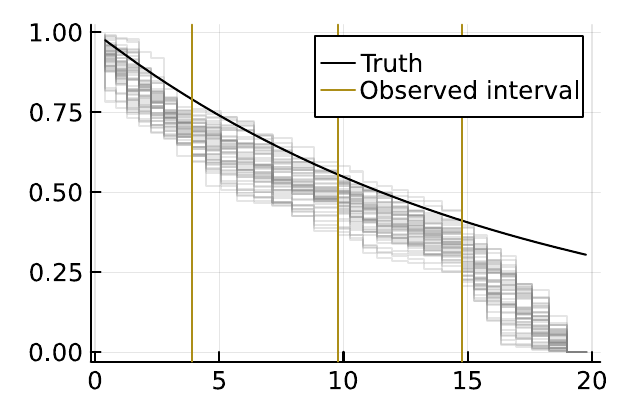}
    \caption{Plot of 50 posterior simulated  survival functions for three
      cohorts (thin grey lines) versus the simulation truth (black
      line).
      The vertical lines mark the observed median and confidence
      interval limits. }
    \label{fig:surv_est_sim}
\end{figure}
Finally, Figure \ref{fig:surv_est_sim} illustrates $p(P_i \mid data)$ by way of
showing 50 posterior draws for the corresponding survival function for
three randomly selected cohorts.
The medians for the estimated curves are close to the true medians and
uncertainty around the curve increases away from the reported interval
boundaries $(\ell_i, m_i, h_i)$ (which are marked by vertical lines).
Keep in mind that the data only informs about the counts in the intervals
between the vertical lines. Consequently posterior inference
can only recover the probability mass for each interval, but not
possibly any finer details. This is most evident in the right most
intervals. The data provides no information about the right tail.

\section{Results}

We fit the proposed BNPMA model to the data presented in Section 2. 
In levels $m=1,2$ of the nested partition tree we
specify the base measure of the
DP prior  \eqref{eq:dep} as a normal distribution,
$\text{N}(0,\sigma^2_m I)$.
Recall that $G_{\eps 0}$ is the prior for the
regression coefficients $\beta_{\eps 0}^{(s)}$, justifying the zero
mean as a symmetric prior. Also,
keeping in mind the nature of the covariates as categorical
indicators, we use $\sigma^2_m = 2^{m+1}$ to define a vague prior.

We implemented the BNPMA with the two (categorical) covariates
of therapeutic agent and tumor type, both coded as multiple binary
indicators.

\paragraph*{Estimating $P_i$.}
We implemented posterior MCMC simulation as proposed before.
After a burn-in of 49,000 iterations, over another 1000 iterations,
$j=1,\ldots,1000$, 
we evaluated the imputed event time distributions $P_i$ for each
iteration.
Figure \ref{fig:surv_52} illustrates estimation of $P_i$ by showing
the corresponding survival function for a randomly selected cohort.
Note the multimodality (in the implied event time distribution),
justifying the use of the nonparametric prior for $P_i$.
\begin{figure}[h]
    \centering
    \includegraphics[width = 0.4\linewidth]{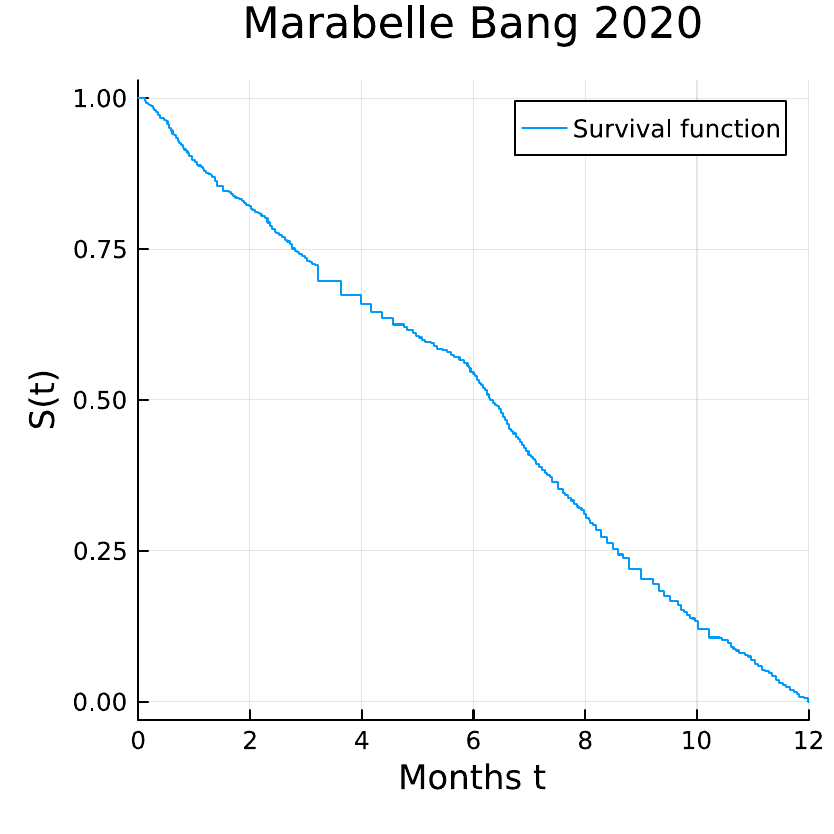}
    \caption{Posterior estimated $P_i$ for a randomly chosen cohort,
      shown as a survival function. Note the apparent changepoint (around
      6). }
    \label{fig:surv_52}
\end{figure}

As in the simulation study we then
evaluated the implied medians $M_i\jj$.
For a study $s$ with marker-positive and marker-negative cohorts $\ip$
and $\im$, let $D_s=\log\{M_{\ip}/M_{\im}\}$ define the difference of
log median event times, and let $D_s\jj$ denote the same under the
parameters imputed in iteration $j$.
Summaries of $D_s\jj$ report the effect of marker-status.
Averaging over all studies we find $1/S \sum p(D_s > 0 \mid data)=.743$, 
This is the posterior probability of larger median event time for
marker-positive than for marker-negative patients.

Further, we compare posterior inference on $D_s$  with the
reported confidence intervals in the original studies. We use 95\%
posterior credible intervals. The comparison is shown in Figure
\ref{fig:forest_compara}. Overall, the credible intervals
report similar effects as the reported confidence intervales, but are
shorter due to the sharing of strength in the BNPMA, and outliers are shrunk towards an overall mean.

\begin{figure}[bt]
  \begin{tabular}{cc}
    \includegraphics[width =
    0.47\linewidth]{Figures/forest_median.pdf} 
    & \includegraphics[width =
      0.47\linewidth]{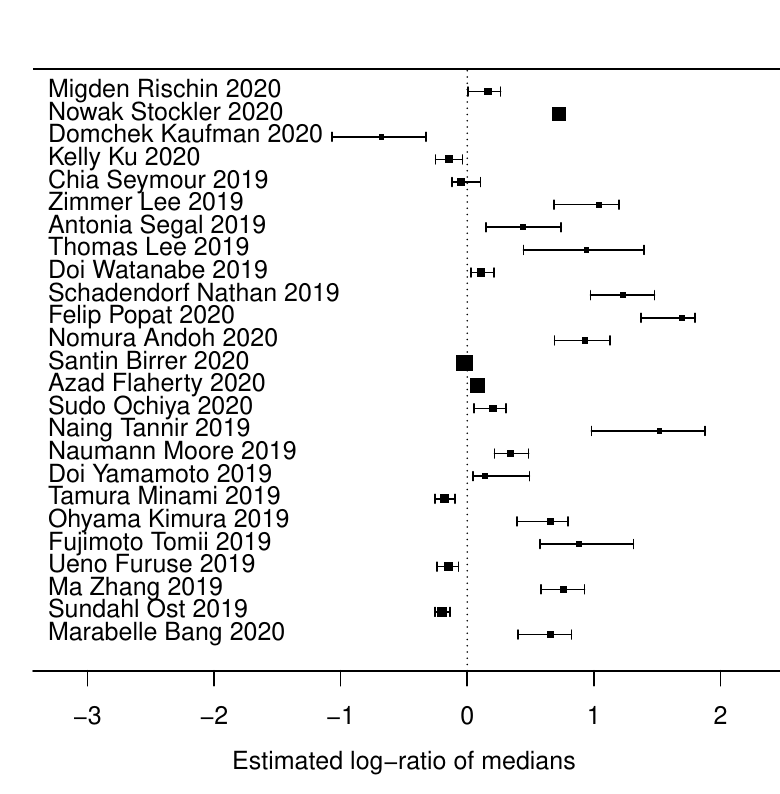}\\
    (a) reported $(\ell_i,m_i,h_i)$ & (b) posterior estimates
  \end{tabular}                    
  \caption{Forest plots of the observed (left) and estimated (right)
    intervals for the median effect of the therapeutic agents. Note
    the difference in scales between both plots.} 
    \label{fig:forest_compara}
\end{figure}

\paragraph*{Meta-regression and posterior predictive inference.}
For meta-regression we are interested in the effect of tumor types or
treatments on $D_s$. 
Since the model does not include any parameters that are
explicitly interpretable as covariate effects on $D_s$ (the $\beta_s$
are regression coefficients for the logit splitting probabilities), we need to
improvise and proceed as follows.
To report covariate effects for categorical covariates, such as agent
and tumor type, we combine $D_s\jj$ across all
studies $s$ with the corresponding value of the covariate. We then
report the resulting distribution of $D_s\jj$ as posterior
distribution for the effect of a covariate value of interest.
We show the summaries for some covariates of interest in 
Figure \ref{fig:sign}. 
For the frequentist analysis we implement a
meta-regression using the \texttt{R} package \texttt{metamedian}. For
the BNPMA results we reported 95\% credible intervals for the difference of
medians for studies with the covariate
vs without it. Note that the posterior intervals for the
difference of log medians under the BNPMA are generally narrower than under the frequentist method, most likely due to the information
sharing between different studies. Also, the clustering allows for
multimodality. According to these results, we find the main effects of
durvalumab, pembrolizumab, patients with melanoma and with the
overflow category of tumors "other" to be non-signficant; in contrast with
the classical method, which reports only nivolumab and melanoma as non
significant.

\begin{figure}[hpbt]
    \centering
    \includegraphics[width = 0.5\linewidth]{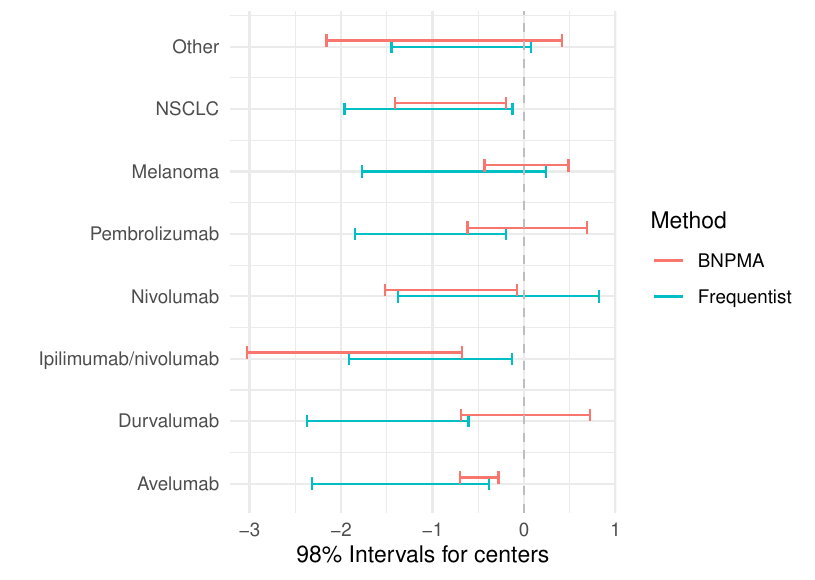}
    \caption{95\% credible (confidence) predictive intervals for the
      difference of log medians induced by different covariates, namely tumor types and treatments. }
    \label{fig:sign}
\end{figure}

We also perform a two-way interaction analysis. For this we evaluate
the predictive intervals for future studies with the observed
tumor-treatment combinations, shown in Figure \ref{fig:inter}.

We again obtain a similar conclusion as under the frequentist method,
with only four disparities. The main one concerning the joint effect of
nivolumab and "other", as the signs of the effect in the difference of
medians are opposite; however, note that the other
interactions with "other" are skewed to the negative axis, suggesting that
result is most likely due to the estimated dependence structure
between effects.

\begin{figure}[htpb]
\centering
  \begin{tabular}{cc}
    \includegraphics[width = 0.34\linewidth]{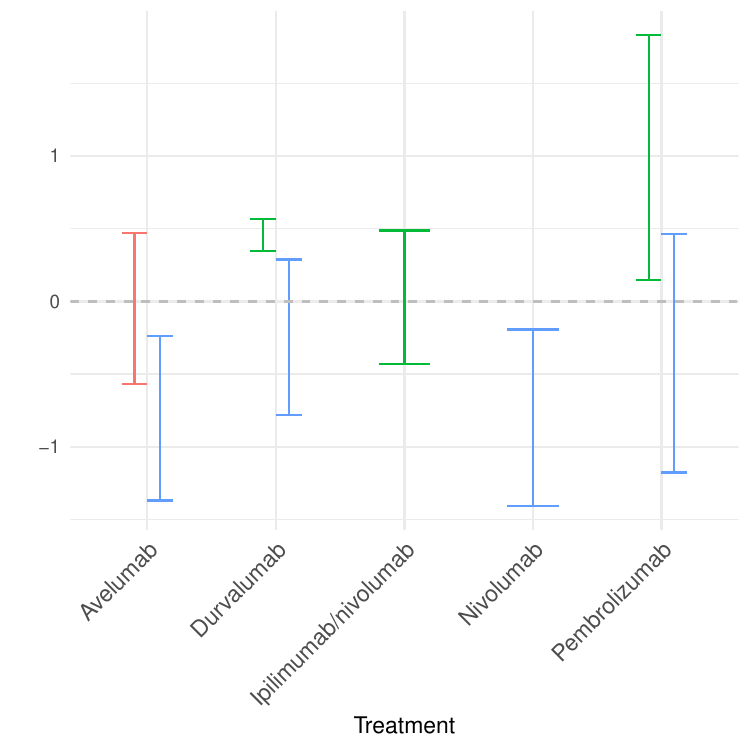}
    & \includegraphics[width =
      0.34\linewidth]{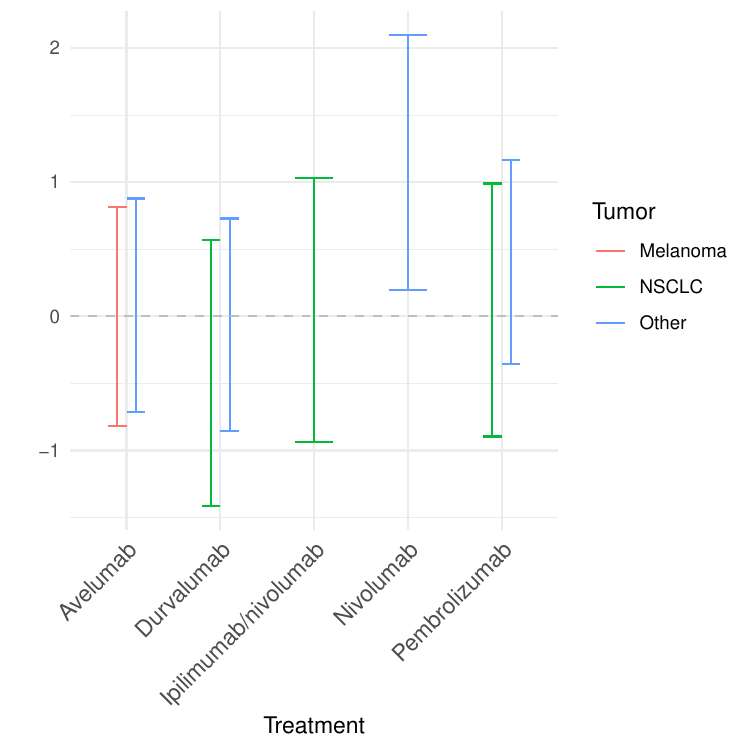}
   \end{tabular}
    \caption{95\% credible (confidence) predictive intervals for the
      second order interaction for the indicated tumor and treatment
      levels.
      The left panel shows the results under the BNPMA; the
      right panel shows the same under classical meta-regression using
      a random effect model.}
    \label{fig:inter}
\end{figure}

For posterior predictive inference for specific tumor types or
treatments  we proceed similarly as for
meta-regression.
\begin{figure}[htb]
\centering
    \includegraphics[width = 0.25\linewidth]{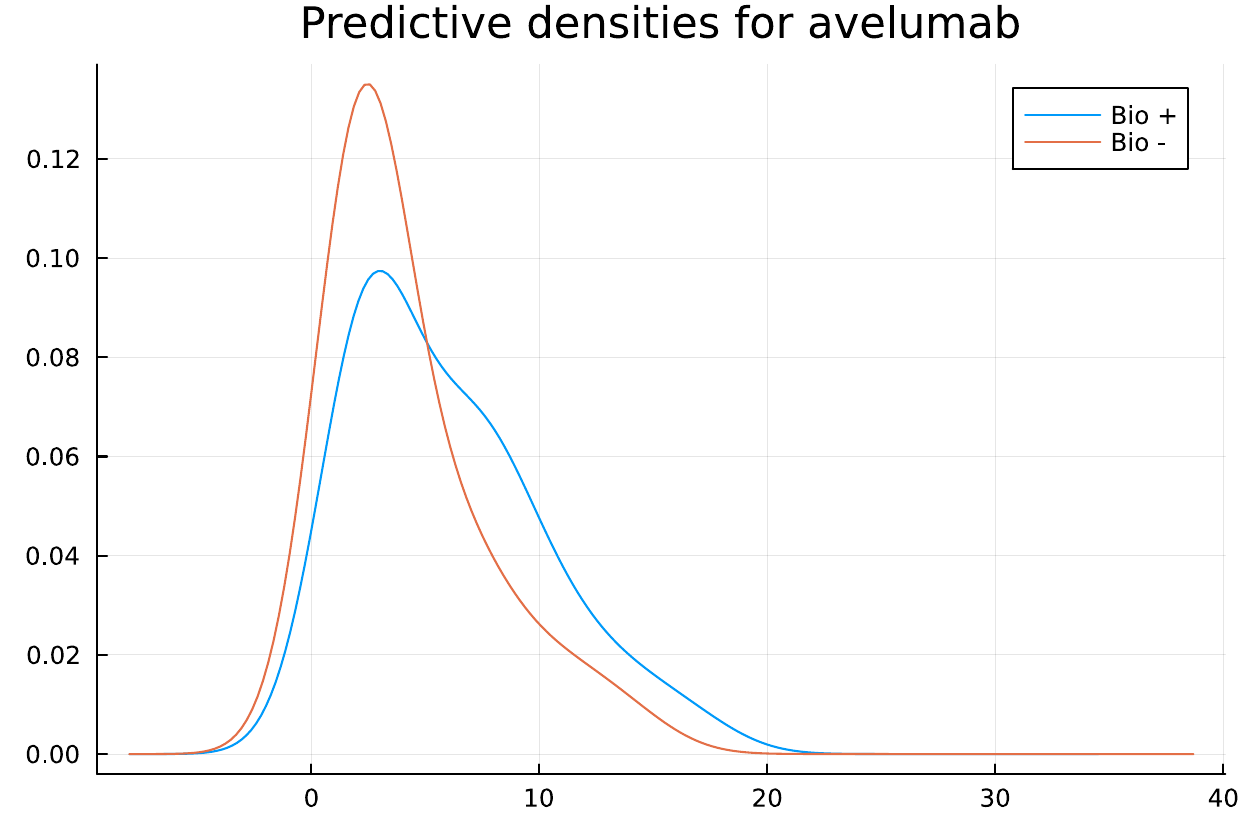}
    \includegraphics[width = 0.25\linewidth]{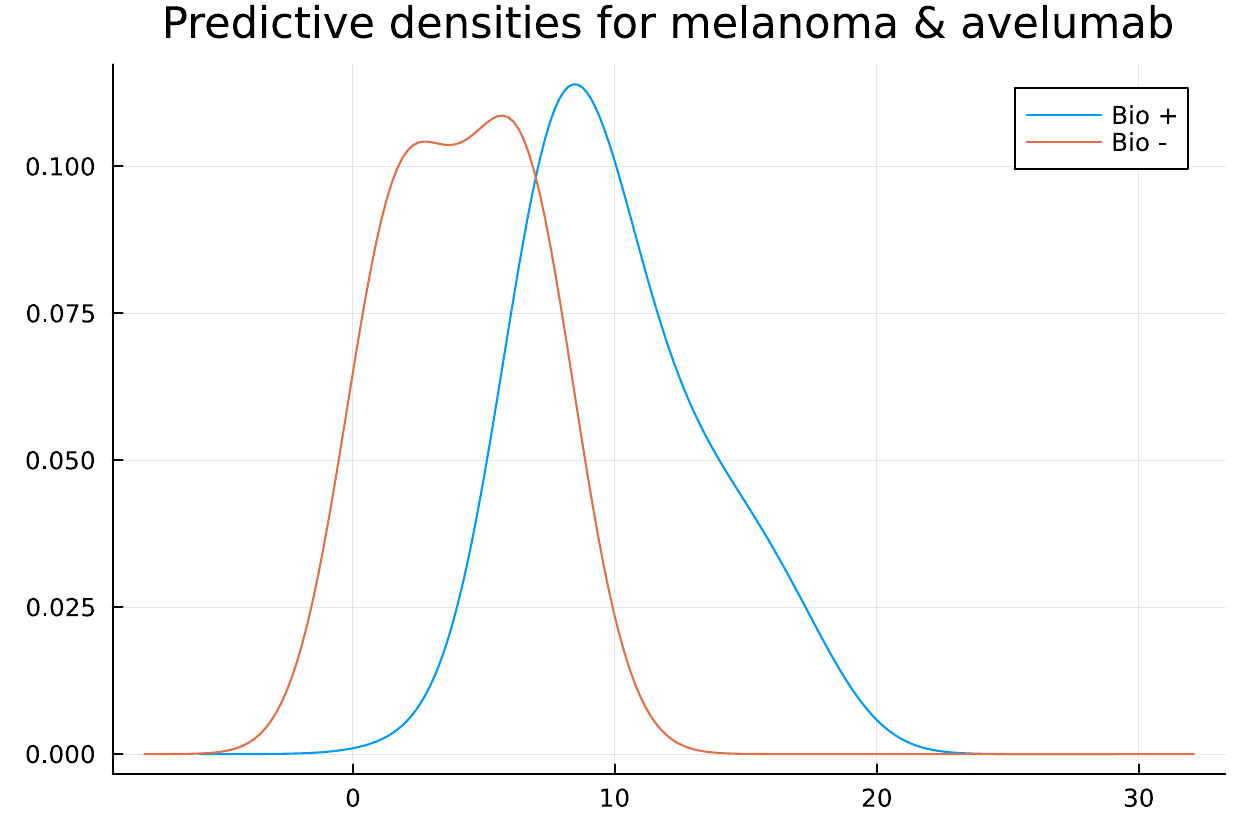}
    \includegraphics[width = 0.25\linewidth]{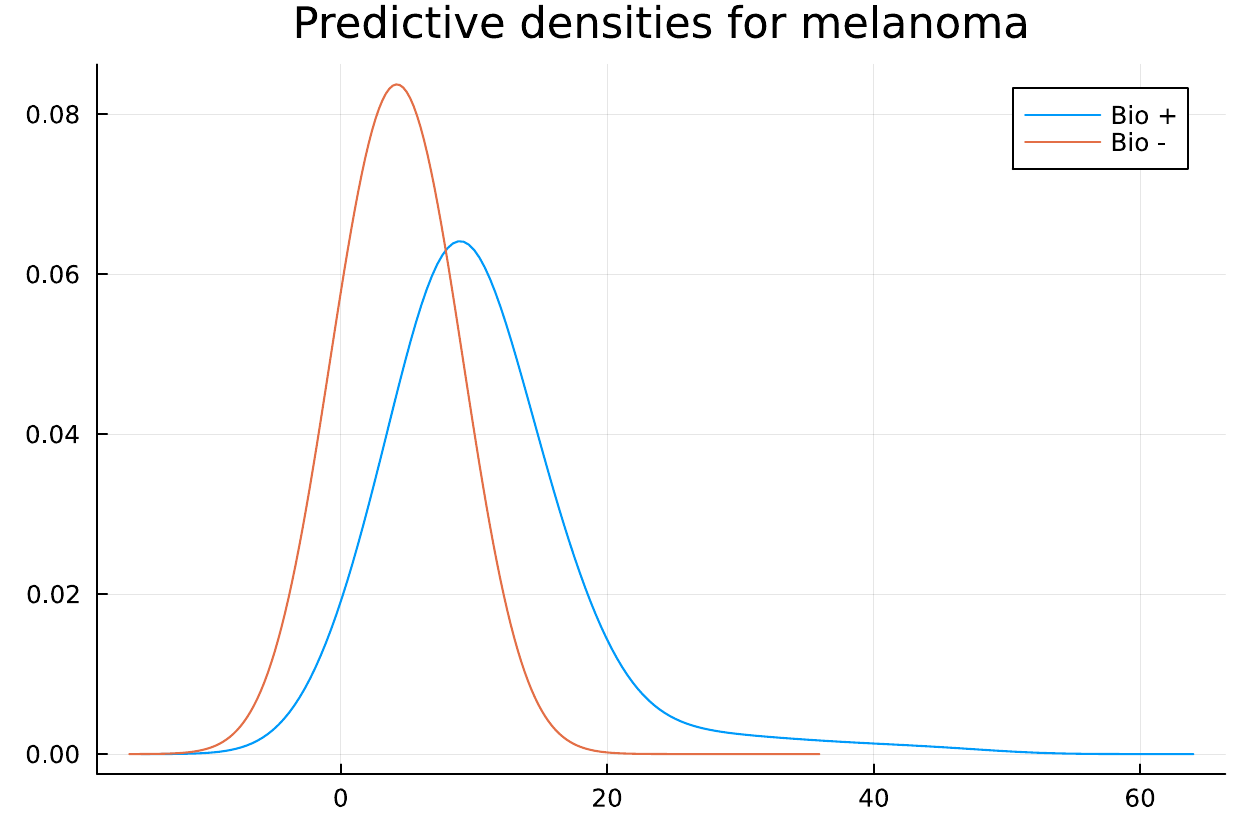}
    \vspace{-5pt}
    \caption{Estimated distributions of the PFS for biomarker-positive and negative patients treated with avelumab, regardless of the type of cancer; with melanoma, regardless of the treatment; and with melanoma treated with avelumab.}
    \vspace{-5pt}
    \label{fig:preds}
\end{figure}
For example, to report a predictive distribution of event
times, PFS in our case, for melanoma we average imputed $P_i$ across
all cohorts with melanoma patients. 
Figure \ref{fig:preds} shows predictive densities for the most
common tumor type and treatment: melanoma and avelumab. The smoothed densities imply a stronger effect of the said
immunothreapy
on biomarker-positive melanoma patients than for biomarker-negative
patients.

\paragraph*{Study heterogeneity -- clustering.}
Finally we report results on the inference for study heterogeneity,
i.e., on the random partition of studies induced by the DP mixture
prior on $\beta_{\epsb0}\ss$. See Section 2 in the supplement for details on how we determine a point estimate for a random partition.

Figures \ref{fig:cluster} and
\ref{fig:clusterfreq} show summaries of the estimated partition of
studies. 
The three largest groups show a model split
 according to different characteristics of the trials; for
instance, clusters 1 and 2 mainly have trials with avelumab and
pembrolizumab, but the first one excludes breast cancer for melanoma
and NSCLC while the second one includes does the opposite.

\begin{figure}[htbp]
    \centering
\includegraphics[width=0.28\linewidth]{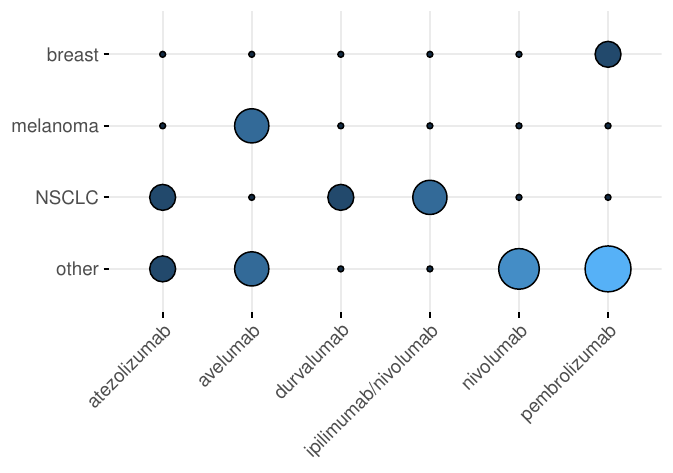}
    \includegraphics[width=0.28\linewidth]{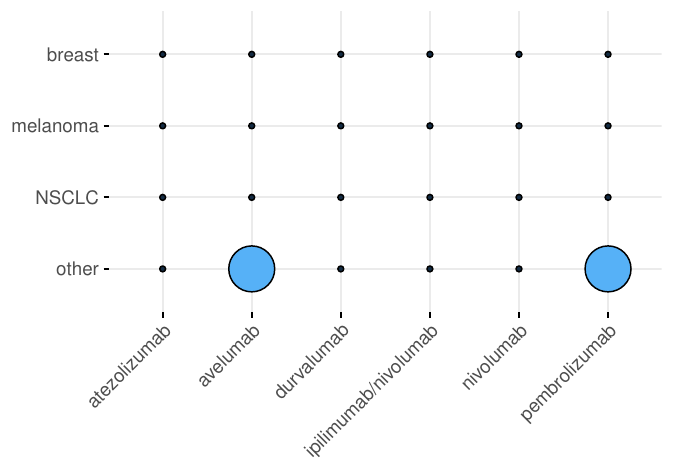}
    \includegraphics[width=0.28\linewidth]{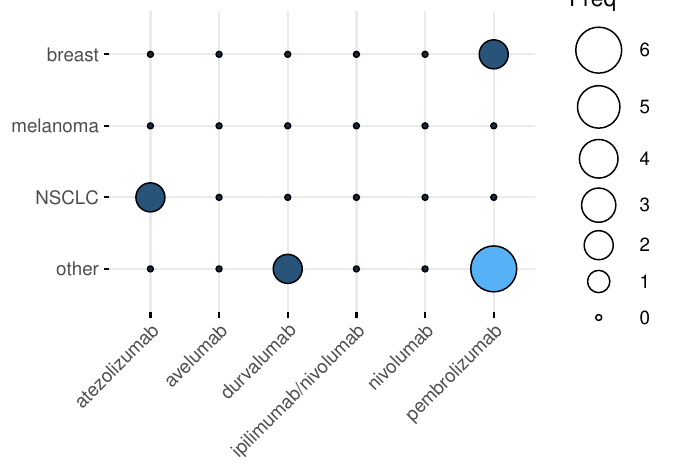}
    \caption{Contingency tables for treatment and tumor type for the three largest clusters.}
    \label{fig:clusterfreq}
\end{figure}

\begin{figure}[htbp]
    \centering
    \includegraphics[width = 0.48\linewidth]{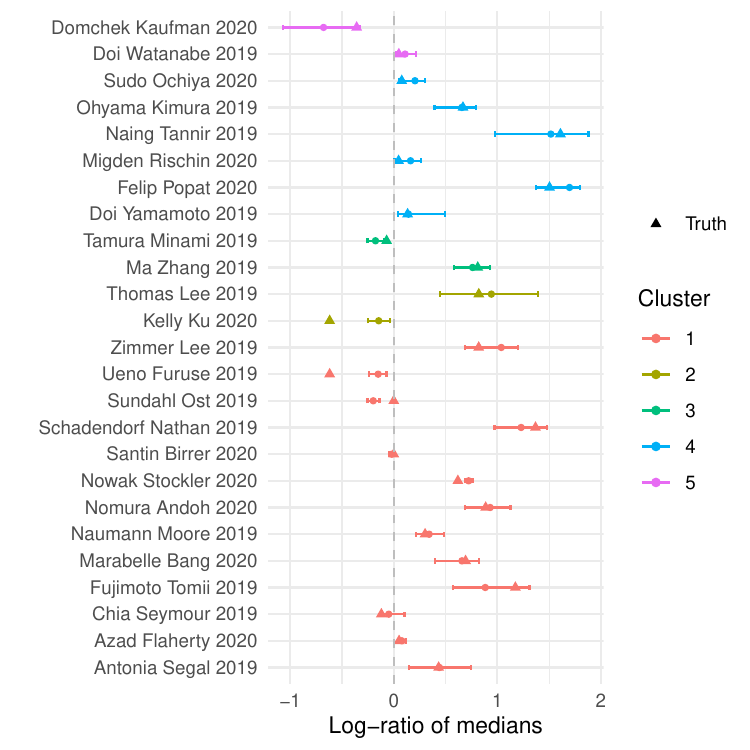}
    \caption{Posterior forest plot. Each color represents one different cluster. Cluster 3 was excluded since it included studies with small sample sizes the \texttt{metamedian} package could not use.}
    \label{fig:cluster}
\end{figure}

\section{Conclusion}

We have introduced a novel non-parametric Bayesian approach for
meta-analysis with event time endpoints. The model naturally
accommodates heterogeneity across included studies and allows
inference without restrictive parametric assumptions.
Using a non-parametric Bayesian mixture of LDTP models allows for flexibilty in the regression on study-specific
covariates, while still taking advantage of the parsimonious nature of
the LDTP. Several alternative extensions of the LDTP are possible, including, for example, the use of more general dependent tailfree processes to link models across studies (or cohorts), as in \cite{poli_multivariate_2023} or in \cite{jara_class_2011}.

Limitations of the proposed approach include the dependence of
inference on the unknown distributions on the chosen partition
sequences, including a lack of smoothness in the estimated densities
at the partition boundaries. Reporting survival functions, as is
customary for event time data, this is less of a problem.
Another limitation is the focus on event time outcomes. The model is
not appropriate, for example, for binary outcomes like tumor response,
but could be used without modification for any other continuous
outcome. This is related to another limitation of the model by not
specifically accommodating censoring. Instead inference hinges only on
the reported point and interval estimates of median event times.
The only practical constraint is that other continuous outcomes might
not commonly be summarized by estimates of median outcomes.

\clearpage
\bibliographystyle{apalike}
\bibliography{main_arxiv}
\clearpage
\end{document}